\pdfoutput=1
\documentclass[a4paper,UKenglish,thm-restate]{lipics-v2019}

\title{Barrington Plays Cards: The Complexity of Card-based Protocols} 

\titlerunning{Barrington Plays Cards} 

\author{Pavel Dvořák}{Charles University, Prague, Czech Republic}{koblich@iuuk.mff.cuni.cz}{}{}

\author{Michal Koucký}{Charles University, Prague, Czech Republic}{koucky@iuuk.mff.cuni.cz}{}{}

\authorrunning{P. Dvořák and M. Koucký} 

\Copyright{Pavel Dvořák and Michal Koucký} 



\category{} 

\relatedversion{} 

\supplement{}

\funding{The authors were supported by Czech Science Foundation GA{\v C}R (grant \#19-27871X).}

\acknowledgements{We thank Václav Blažej for suggesting how to use face-up cards to emulate red and blue card backs.}

\nolinenumbers 

\hideLIPIcs  

\EventEditors{John Q. Open and Joan R. Access}
\EventNoEds{2}
\EventLongTitle{42nd Conference on Very Important Topics (CVIT 2016)}
\EventShortTitle{CVIT 2016}
\EventAcronym{CVIT}
\EventYear{2016}
\EventDate{December 24--27, 2016}
\EventLocation{Little Whinging, United Kingdom}
\EventLogo{}
\SeriesVolume{42}
\ArticleNo{23}

\newcommand{\NC}{{\sf NC^1}}
\newcommand{\BP}{{\sf PB}}
\newcommand{\PBP}[1]{#1\text{-}{\sf PBP}}
\newcommand{\CS}{{\sf SP}}
\newcommand{\NL}{{\sf NL}}
\newcommand{\LOG}{{\sf L}}

\newcommand{\cz}{\clubsuit}
\newcommand{\co}{\heartsuit}
\newcommand{\ce}{\times}
\newcommand{\cq}{\,?\,}
\newcommand{\SEL}{\textit{SEL}}
\newcommand{\CONN}{\textit{CONN}}

\begin{document}
\maketitle

\begin{abstract}
In this paper we study the computational complexity of functions that have efficient card-based protocols.
Card-based protocols were proposed by den Boer~\cite{boer90} as a means for secure two-party computation. 
Our contribution is two-fold: We classify a large class of protocols with 
respect to the computational complexity of functions they compute, and we propose other encodings of inputs 
which require fewer cards than the usual 2-card representation.
\end{abstract}

\section{Introduction}
Card-based protocols were proposed by den Boer~\cite{boer90} as a means for secure two-party computation. In this scenario, we have two players --- Alice and Bob --- who hold inputs $x$ and $y$ respectively. Their goal is to securely compute a given function $f$ on those inputs. By secure computation, we mean that the players learn nothing from observing the computation except for what is implied by the output $f(x,y)$.  Den Boer 
introduced a model where the inputs $x$ and $y$ are encoded by a sequence of playing cards and the players operate on the cards to compute the function. They can use additional cards for computation. In particular, den Boer showed how to securely compute AND of two bits using five cards in total. 

Niemi and Renvall~\cite{niemi97} extended the results to show how to compute arbitrary Boolean function $f$.
They represent each input bit by two face-down cards: 1 is represented as $\co\cz$, and 0 as $\cz\co$. They use the same representation throughout the computation and show how to securely compute AND of bits encoded in this \emph{2-card} representation with the output being encoded in the same way. Since NOT can be obtained by swapping the two cards representing a given bit this allows them to compute any function. They can evaluate any Boolean circuit on the inputs by a protocol of length proportional to the size of the circuit and using a number of auxiliary cards that corresponds to the \emph{width} of the circuit.

Nishida et al.~\cite{nishida15} reduced the number of auxiliary cards to 6 for any Boolean function. For most functions, the protocol will be of exponential length as it essentially evaluates the DNF of $f$. Several other works studied the number of cards necessary for computing various elementary functions such as AND and XOR~\cite{danny17,koch15,koch18,kastner17,nishida15,stiglic01,mizuki09,mizuki12,mizuki14a,koch15,abe18,koch18}.

Motivated by the question what can be efficiently computed by such protocols and how many cards one needs to compute various functions, in this work, we investigate \emph{efficient} protocols that are protocols of polynomial length. Our contribution is two-fold: We classify a large class of protocols with respect to the computational complexity of functions they compute, and we propose other encodings of inputs which require fewer cards than the 2-card representation. We summarize our results next:

\smallskip\noindent{\bf 1.} We show that \emph{oblivious} protocols of polynomial length that do not modify their input (they are \emph{read-only}) and use only a constant number of auxiliary cards compute precisely the functions in $\NC$, the class of functions computed by Boolean circuits of logarithmic depth. (Alternatively, $\NC$ is the class of functions computed by Boolean formulas of polynomial size.) By oblivious protocol we mean a protocol whose actions depend only on the current visible state.

\smallskip\noindent{\bf 2.} Oblivious read-only protocols of polynomial length with a logarithmic number of auxiliary cards correspond to the class of functions
computable by polynomial-size branching programs. (This class is also known as $\LOG/poly$, the non-uniform version of deterministic log-space.)

\smallskip\noindent{\bf 3.} We also investigate protocols that use a constant number of auxiliary cards but are allowed to use the cards representing the input for their computation provided that they guarantee that by the end of the computation the input will be restored to its original value. We show that such protocols can compute functions that are believed to be outside of $\NC$. For example, they can compute languages that are complete for $\NL$, the non-deterministic log-space. Hence, read-only protocols are presumably weaker than protocols that may modify their input.

\smallskip\noindent{\bf 4.} We study alternative encodings of inputs that are more efficient that the 2-card encoding. We look at 1-card encoding where 1 is represented by $\co$ and 0 by $\cz$.  In this encoding, Alice and Bob need only one card per bit to commit the bit. We show similar complexity results for this encoding as for the 2-card encoding: read-only protocols with a constant number of auxiliary cards are $\NC$, with a logarithmic number of cards it is the non-uniform log-space, and if we allow using the input cards for computation we get potentially more powerful protocols.

\smallskip
A disadvantage of the 1-card protocol is that it still needs a supply of $n$ cards $\co$ and $n$ cards $\cz$ to represent any $n$-bit input. 
Although, if one restricted his attention to inputs that contain the same number of 1's and 0's, it would suffice to have $n/2$ cards $\co$ and $n/2$ cards $\cz$. Such inputs form a substantial fraction of all $n$-bit inputs, they are $\Theta(2^{n}/\sqrt{n})$ many.

\smallskip\noindent{\bf 5.} We propose a new $1/2$-card encoding which requires only $n/2$ cards $\co$ and $n/2$ cards $\cz$ to represent any $n$-bit input.
The $1/2$-card encoding is obtained from the 2-card encoding by removing from each pair of cards one card, in total one half of the $\co$-cards and one half of the $\cz$-cards. There is an empty space left instead of each removed card. There is a way for each player to encode his input so that the other player learns no information about the opponent's input. We show that using this encoding we can simulate any read-only protocol that uses 2-card encoding.
Hence, any $\NC$ function on $n$ bits can be securely computed using only $n+O(1)$ cards, counting also the input cards. We do not know how to securely perform protocols for $1/2$-card encoding that would modify their input.

\subsection{Previous Work}
As mentioned above, a study of card-based protocols was started by den Boer~\cite{boer90} who introduced a secure 5-card protocol for computing AND.
However, this protocol does not produce output in a face-down 2-card format, thus it can not be used for designing protocols for arbitrary function.
Since then a lot of work was done in improving AND protocols and other primitive functions.
Cr{\'e}peau and Kilian~\cite{crepeau94} provided a 1-party card-based protocol where the player can pick a random permutation $\pi$ with no fixed point and the player has no information about $\pi$.
Their technique can be used for designing a 2-party computation of a general function.
Niemi and Renvall~\cite{niemi97} introduced an AND protocol, which takes two bits $b_1, b_2$ represented in the 2-card format as input and outputs two cards which represent $b_1 \wedge b_2$ in face-down 2-card format.
They also introduced a protocol for copying a bit in the 2-card format, which is used during simulation of circuits.
Their protocols with a protocol for NOT (which is trivial) can be used for computing any Boolean function $f$ and the number of used cards is at most linear in the size of a circuit (using AND and NOT gates) computing $f$.
They also introduced a protocol to copy a single card with almost perfect security -- the card suit is revealed only with a small probability.
Such protocol cannot exist with perfect security as was proved by Mizuki and Shizuya~\cite{mizuki14}.
The copying and AND protocols were further improved and simplified in~\cite{stiglic01,mizuki09,mizuki12,mizuki14a,koch15,abe18,koch18}.

Nishida et al.~\cite{nishida15} proved that any Boolean function $f: \{0,1\}^n \times \{0,1\}^n \to \{0,1\}$ can be computed with $4n$ cards encoding the input and 6 additional helping cards.
Mizuki~\cite{mizuki16} proved that $2n + 2$ is needed to compute AND of $n$ bits.
Francis et al.~\cite{danny17} provided protocols and tight lower bounds on the number of cards needed for computing any two-bit input two-bit output function.
Other lower bounds for AND of 2 bits and of $n$ bits in various regimes were provided by Koch et al.~\cite{koch15, koch18} and by Kastner et al.~\cite{kastner17}.

The security of card-based protocols is provided by shuffling the cards --- one player shuffles the cards (applies some random permutation to them) in a way so that the other player has no information about the new order of shuffled cards.
Koch et al.~\cite{koch15} provided a 4-card AND protocol.
However, they used a non-uniform distribution for picking a random permutation, which is difficult to perform by humans.
Nishimura et al.~\cite{nishimura16} suggested an ``easy-for-human'' procedure how to apply a shuffling permutation picked from a non-uniform distribution using envelopes.

One can distinguish two types of attack.
\begin{enumerate}
 \item Passive: honest-but-curious player -- she follows the protocol but she wants to retrieve as much information as possible about the other player input.
 \item Active: malicious player -- she can deviate from the protocol.
\end{enumerate}
Koch and Walzer~\cite{koch21} proved that if a passive-secure protocol $\Pi$ uses only uniform closed shuffles (each shuffling permutation is picked uniformly from some permutation group) then the protocol $\Pi$ can be transformed into an active-secure protocol.

\section{Preliminaries}

\subsection{Card-based Protocols}
\label{sec:PrelimProtocol}
In this section we define \emph{card-based protocols} which securely compute some Boolean function on a joint input of Alice and Bob.
Alice gets an input $x \in \{0,1\}^n$ and Bob gets an input $y \in \{0,1\}^n$, and their goal is to compute $f(x,y)$ for some function $f: \{0,1\}^n \times \{0,1\}^n \to \{0,1\}$, while not revealing anything about their input to the other player.
The protocol proceeds first by Alice and Bob committing their input into a sequence of cards, and then operating on the cards together with some auxiliary cards. At the end of the protocol, the players learn the output $f(x,y)$. 

In this section we consider the usual \emph{2-card encoding} of the input, where each input is represented as a sequence of cards, two cards per bit: value $1$ is represented by $\co\cz$ 
and value $0$ is represented by $\cz\co$ where the cards are put face-down on the table. Hence each player needs $2n$ cards to commit his input.
All the cards have the same back, say blue. 
In the beginning, face-down cards representing the player inputs are in front of the players.
Between them, there is a \emph{deck} of $s$ prescribed auxiliary cards of $\co$ and $\cz$.
There is available some empty space on the table to operate with the cards.
We assume that the cards are placed on the table in some specific positions (locations), numbered $1,\dots,m$, where:
\begin{itemize}
 \item $1,\dots,2n$ are positions of Alice's input cards,
 \item $2n+1,\dots,4n$ are positions of Bob's cards,
 \item $4n+1,\dots,4n+s$ are the initial positions of the helping cards in the deck,
 \item $4n + s +1,\dots, m$ are initially empty positions.
\end{itemize}
We call the positions $1,\dots,4n$ as the input positions and the remaining positions as the \emph{work space}.
We say a position is \emph{occupied} if there is a card on it, otherwise, it is \emph{empty}. We denote an empty position by $\ce$.
Let $q=m-4n$ denote the amount of the work space. We assume $q=O(s)$.
Thus, there are $4n + s$ cards on the table and $4n + q = m$ positions.

The players can move their input cards and cards from the deck to the work space and back.
Formally, the basic \emph{actions} which can be executed by the players are:
\begin{description}
 \item[Move($p, i, j$)]: The player $p$ moves a card from the position $i$ to position $j$. 
  \item[Shuffle($p, T, \Gamma$)]: The player $p$ applies a random permutation from $\Gamma$ to the cards on the table on positions $T \subseteq \{4n+1,\dots,m\}$. 
  \item[Turn($p,i$)]: The player $p$ turns the $i$-th card on the table face-up if it is face-down, and vice versa.
\end{description}

The protocol specifies which action to take next based on the sequences of visible states seen on the table so far. The current visible state of the table is what an external viewer could observe, that is which positions are currently occupied and what is the top of each card laying on the table. If there are $c$ distinct cards then there are at most $(c+2)^m$ distinct visible states.
Hence, based on the sequence of visible states from the beginning of the game the protocol specifies which action to take next or whether to end. In the end, the protocol specifies which cards represent the output of the run of the protocol. (They might be face-down.) We say the protocol $\Pi$ \emph{computes} a function $f: \{0,1\}^n \times \{0,1\}^n \to \{0,1\}$ if for all inputs $x,y \in \{0,1\}^n$, on the inputs $x$ and $y$ the protocol outputs $f(x,y)$.
The length of the protocol is the maximum number of actions executed by the protocol over all inputs $(x,y) \in \{0,1\}^n \times \{0,1\}^n$
and all possible outcomes of shuffling. We say that a protocol is \emph{oblivious} if the action executed next depends only on the current visible state and the number of actions taken so far.

The shuffling operation provides randomness for the execution of the protocol. Hence, the sequence of visible states the protocol passes through is a random variable. We will say that a protocol is \emph{secure} if for any pair of inputs $(x,y)$ and $(x',y')$ to Alice and Bob, 
where $f(x,y)=f(x',y')$, the distribution of the sequence of visible states of the protocol on inputs $(x,y)$ and $(x',y')$ is the same.
Notice, that this implies that neither of the players learns anything about the input of the other player except for what is implied by $f(x,y)$.

Often we will be interested in protocols that provide their output encoded in face-down cards. In such a scenario we will require for the security of the protocol that the distributions of visible states during the protocol will be identical for all input pairs $(x,y)$.

We say the protocol is \emph{robust} if a cheating player, that is a player who deviates from the protocol, is either caught by reaching an invalid visible state (where cards have unexpected values or positions) or the distribution of visible states does not leak any information about the other player input except for what would be leaked by honest players. In particular, if say Bob is cheating and Alice is honest, for a robust protocol we require that for any input $x$ of Bob and any two inputs $y,y'$ to Alice, where $f(x,y)=f(x,y')$, the distribution of the sequence of visible states during the game on inputs $(x,y)$ and $(x,y')$ is the same. We will be designing only robust oblivious protocols.

We say the protocol is \emph{read-only} if the value of cards placed on the input positions $1,\dots,4n$ is always the same whenever a position is occupied. 

Let $s$-$\CS$ be the class of function families $\{ f_n: \{0,1\}^n \times \{0,1\}^n \to \{0,1\} \}_{n\ge 0}$ for which we have a sequence of secure read-only oblivious protocols, one for each $n$, which are of length polynomial in $n$, with deck size $s$ and work space size $2s$. (At the beginning the first $s$ work space positions are occupied by the deck of cards, and the remaining $s$ positions are empty). We might allow $s$ to be a function of $n$. We define $\CS= \bigcup_{s\ge 1} s\textit{-}\CS$.
That is a function belongs to $\CS$ if it has polynomial length protocols which use a constant number of auxiliary cards and constant size work space.

\subsection{Branching Programs}
A \emph{branching program} $B$ for a Boolean function $f: \{0,1\}^n \to \{0,1\}$ is defined as follows.
It consists of a directed acyclic graph $G$ such that each vertex has out-degree either 2 or 0.
The set of edges $E$ of the graph $G$ is split into two sets, zero-edges $E_0$ and one-edges $E_1$, in such a way that every vertex $v$ of out-degree 2 is incident to exactly one outgoing zero-edge and exactly one outgoing one-edge.
Each vertex of out-degree 2 is labeled by an index $\ell \in [n]$.

A branching program $B$ is \emph{layered} if the vertices are partitioned into layers $L_1,\dots,L_d$.
The edges go only from a layer $L_i$ to a layer $L_{i+1}$ (for all $i < d$).
Vertices of out-degree 0 are exactly vertices in the layer $L_d$.
The number of layers $d$ is the \emph{length} of $B$ and the \emph{width} $w$ of $B$ is the maximum size of its layers, i.e., $w = \max_i |L_i|$.

A layered branching program is \emph{oblivious} if vertices in the same layer have the same label.
A branching program is a \emph{permutation branching program} if each layer has exactly $w$ vertices and for every two consecutive layers $L_j$ and $L_{j + 1}$ zero-edges and one-edges form matching $M^0_j$ and $M^1_j$, respectively.
We can view the matchings $M^0_j$ and $M^1_j$ as two permutations $\pi^0_j, \pi^1_j: [w] \to [w]$.
Note that we can rearrange all layers such that all permutations $\pi^0_j$ are identities.

One vertex of in-degree 0 is an \emph{initial} vertex $\bar{v}$.
Some vertices of out-degree 0 are denoted as \emph{accepting} vertices.
The computation of a branching program $B$ on an input string $x \in \{0,1\}^n$ proceeds as follows.
It starts in the initial vertex $\bar{v}$ which is the first \emph{active} vertex.
Suppose $v$ is an active vertex and $\ell \in [n]$ is the label of $v$.
If the out-degree of $v$ is 2, then the next active vertex is determined by the zero- or one-edge according to the value of $x_\ell$.
More formally, let $e = \{v,v'\}$ be the edge in $E_{x_\ell}$.
Then, the vertex $v'$ is the new active vertex.
We repeat this procedure until a vertex $u$ of out-degree 0 is reached.
An input $x \in \{0,1\}^n$ is accepted if and only if $u$ is an accepting vertex.
The branching program $B$ computes a function $f: \{0,1\}^n \to \{0,1\}$ if it accepts exactly those $x \in \{0,1\}^n$ such that $f(x) = 1$.

The class of functions {\BP} contains all the functions computable by layered branching programs of constant width and polynomial length.
A permutation branching program is \emph{restricted} if it has exactly one accepting vertex $v_\text{acc}$ and exactly one rejecting vertex $v_\text{rej}$ in the last layer $L_d$.
The computation of a restricted permutation branching program ends always in the vertices $v_\text{acc}$ or $v_\text{rej}$ and it accepts an input if it ends in the accepting vertex $v_\text{acc}$.
A class $\PBP{w}$ contains Boolean functions which are computable by restricted permutation branching programs of width $w$ and polynomial length.
We use the famous Barrington's theorem~\cite{barrington89}, which says that constant-width (permutation) branching programs are as powerful as $\NC$-circuits.

\begin{theorem}[Barrington~\cite{barrington89}]
\label{thm:Barrington}
 $\BP \subseteq \NC \subseteq \PBP{5}$.
\end{theorem}

\section{Simulating Branching Programs}
In this section, we prove one of our main theorems that read-only oblivious protocols of polynomial length that use constant work space compute the same functions as polynomial-size constant-width branching programs.

\begin{theorem}\label{thm:OneCard}
$\CS=\NC$.
\end{theorem}

To simulate a branching program by $\CS$-protocol we need an oblivious implementation of copying a bit in the committed 2-card format.
We use a procedure by Stiglic~\cite{stiglic01}.
It is straightforward to implement the procedure to be oblivious.
We include the proof for the sake of completeness.

\begin{restatable}{theorem}{copybit}
\label{thm:CopyBit}
 There is a secure oblivious protocol that takes a bit $b$ in 2-card representation placed in the work space and produces two 2-card copies of the bit in the work space. The protocol needs an auxiliary deck with three cards $\co$ and three $\cz$ with the same back as the input bit.
\end{restatable}
\begin{proof}
\begin{enumerate}
 \item Alice arranges the cards from the auxiliary deck face-up to create the following configuration.
 \[
 \underbrace{\cq\cq}_b \co\cz\co\cz\co\cz
 \]
 \item She turns the last six cards. Both, Alice and Bob, apply a random cyclic shift (denoted by $\langle,\rangle$) to them.
 \[
  \underbrace{\cq\cq}_b \langle\underbrace{\cq\cq}_{b'}\underbrace{\cq\cq}_{b'}\underbrace{\cq\cq}_{b'}\rangle
 \]
  \item They apply a random cyclic shift to the first four cards.
  \[
  \langle \cq\cq\cq\cq \rangle \underbrace{\cq\cq}_{b'}\underbrace{\cq\cq}_{b'}
 \]
 \item She turns the first four cards face-up. 
 \begin{enumerate}
  \item If the sequence is alternating (i.e., $\co\cz\co\cz$ or its shift) then $b = b'$. Thus, the last 4 cards represent two copies of $b$.
  \[
   \co\cz\co\cz \underbrace{\cq\cq}_{b}\underbrace{\cq\cq}_{b}
  \]

  \item Otherwise (i.e., $\co\cz\cz\co$ or its shift) then $b = 1 - b'$. Thus, the last 4 cards represent two copies of negation of $b$. In that case, she switches the fifth with the sixth card and the seventh with the eighth card to represent two copies of $b$ as well.
  \[
   \co\cz\cz\co \underbrace{\cq\cq}_{1 - b}\underbrace{\cq\cq}_{1 - b}
  \]
 \end{enumerate}
Alice and Bob might want to turn over and shuffle the first four left-over auxiliary cards after step 4.
The last four cards represent two copies of $b$ face-down in the 2-card format. To make the protocol oblivious we must implement both 4.a) and 4.b) by the same number of actions. To do so we include additional actions in 4.a) which have no effect such as shuffling a single card. 
\end{enumerate}

It is clear the described protocol is secure.
The only step where they can gain some information about $b$ is Step 4, when Alice turns some cards.
However, the cyclic shifts in Steps 2 and 3 were done by both players.
Thus, Alice reveals the alternating sequence ($\co\cz\co\cz$) in Step 4 with the probability exactly $\frac{1}{2}$ (independently on the value of $b$) even if one of the players would be cheating.
Thus, the protocol is secure.
\end{proof}

To prove $\NC \subseteq \CS$ we use as the first step Barrington's theorem~\cite{barrington89}.
By Barrington's theorem, each function $f : \{0,1\}^n \times \{0,1\}^n \to \{0,1\}$ from $\NC$ can be computed by a polynomial length width-5 restricted permutation branching program. We will build a protocol that simulates the actions of the branching program layer by layer. 
We will keep track of the image of the initial vertex of the branching program. For that, we will use five cards $\co\cz\cz\cz\cz$, where the position
of $\co$ corresponds to the image of the initial vertex (the active vertex), and we will apply the permutations prescribed by the branching program on those five cards.
If the input variable assigned to a particular level of the branching program is set to 1 we are expected to perform the permutation otherwise we are supposed to do nothing, i.e., apply the identity permutation. Any permutation can be decomposed into a sequence of simple transpositions (swaps) so we will use swaps conditioned by the input variable to either permute the five cards or leave them the way they are.
We will implement the following primitive:
Alice and Bob want to conditionally swap two cards $\alpha,\beta$ according to the value of bit $b$ represented in the face-down 2-card form in the work space without revealing the value of $b$. They also want to make sure that if $b = 1$ the swap occurs and if $b = 0$ the swap does not occur.

\begin{theorem}
\label{thm:SwapCards}
 Let  $\tilde{\alpha},\tilde{\beta}$ be two sequences of face-down cards of the same length in the work space, and let $\gamma\,\delta$ be a face-down 2-card representation of $b$ in the work space. 
 There is a secure oblivious protocol such that during the protocol players swap the sequences $\alpha$ and $\beta$ if and only if $b = 1$. The protocol uses two auxiliary cards $\cz$.
\end{theorem}
\begin{proof}
The swapping protocol works as follows.
\begin{enumerate}
 \item Alice rearranges the input cards together with two auxiliary face-up cards $\cz$ as follows:
 \[
   \cz \gamma~\tilde{\alpha}~\cz \delta~\tilde{\beta}
 \]
 Thus, if $b = 0$ we have $\cz \cz \tilde{\alpha} \cz\co \tilde{\beta}$ and if $b = 1$ we have $\cz\co \tilde{\alpha} \cz\cz \tilde{\beta}$. The players do not know which situation are they in.
 \item Both, Alice and Bob, apply a random cyclic shift to the cards, e.g.:
 \[
  \tilde{\alpha}~\cz\delta~\tilde{\beta}~\cz\gamma
 \]
  \item Alice turns the cards $\gamma$ and $\delta$ representing $b$ face-up. She knows what cards to turn, as the cards $\gamma$ and $\delta$ are preceded by $\cz$ face-up.
  At the end she reorders the sequence (keeping the cyclic order) so that $\cz\cz$ are the first cards, e.g.:
  \[
   \tilde{\alpha} \cz\co \tilde{\beta} \cz\cz \rightarrow \cz\cz \tilde{\alpha} \cz\co \tilde{\beta}
  \]
  If $b = 0$ then $\gamma \delta = \cz\co$ and the sequences $\tilde{\alpha}\tilde{\beta}$ are not swapped. 
  On the other hand if $b = 1$ then  $\gamma \delta = \co\cz$ and the sequences are swapped.
\end{enumerate}
Note that the cards in $\tilde{\alpha}$ and $\tilde{\beta}$ are face-down during the whole protocol.
It is also clear that this is a secure and robust protocol, and it can be implemented obliviously.
In Step 3 the cards $\delta$ and $\gamma$ representing the bit $b$ are revealed.
However, because of random cyclic shifts in Step 2 (again done by both players), these cards are in the order $\co\cz$ with probability $\frac{1}{2}$, independently of the value of $b$.
Thus the swapping protocol is secure.
\end{proof}

Now we are ready to prove the first inclusion of Theorem~\ref{thm:OneCard}.

\begin{theorem}
\label{thm:OneCard1}
 $\PBP{5} \subseteq \CS.$
\end{theorem}
\begin{proof}
Let $f: \{0,1\}^n \times \{0,1\}^n \to \{0,1\}$ be a function in $\PBP{5}$. 
Then, the function $f$ can be computed by a branching program $P$ with the following properties:
\begin{enumerate}
 \item Each layer has exactly 5 vertices. The input vertex is the first vertex in the first layer. The computation ends either in the accepting vertex $v_\text{acc}$ or in the rejecting vertex $v_\text{rej}$.
 \item The permutation from each layer $i$ to the layer $i + 1$ corresponding to 0 is the identity.
\end{enumerate}
Alice and Bob represent the first layer as $\co\cz\cz\cz\cz$, each card represents one vertex in the layer.
The card $\co$ represents the active vertex in the layer (the initial vertex in the first layer).
We call these 5 cards the \emph{program cards}.
Alice and Bob put the program cards at the work space and turn them face-down.
Alice and Bob simulate the program $P$ layer by layer. They apply permutations determined by $P$ to the program cards according to the player's input bits.
Suppose we have a representation of the active vertex in the $i$-th layer and we want to calculate the active vertex at the $(i+1)$-th layer.
Without loss of generality the label of the $i$-th layer is Alice's bit $x_\ell$ (otherwise the roles of Alice and Bob are reversed).
Thus, we want to apply some permutation $\rho_i \in S_5$ to the program cards if $x_\ell = 1$ and keep the order of the program cards if $x_\ell = 0$.

We decompose the permutation $\rho_i$ into transposition $\tau_1 \circ \dots \circ \tau_r$.
For $j=1,\dots,r$, Alice will apply the transposition $\tau_j$ to the program cards. She runs the protocol $\Gamma_1$ of Theorem~\ref{thm:CopyBit} to get cards $\gamma,\delta$ representing her bit $x_\ell$ in the work space.
More formally, after the execution of $\Gamma_1$, there are two pairs of cards such that each pair represent the bit $x_j$ in the 2-card format.
She puts one pair back to the input positions, i.e., the protocol $\Pi$ is indeed read-only.
We denote the cards of the second pair as $\gamma$ and $\delta$. Alice will use them for a conditional swap.
She runs the protocol $\Gamma_2$ given by Theorem~\ref{thm:SwapCards} (applied to the cards $\gamma,\delta$ and to the two program cards which should be affected by the transposition $\tau_j$).
That is, Alice swaps the two cards that $\tau_j$ is acting on if and only if $x_\ell = 1$.
After applying this procedure for all transpositions $\tau_1, \dots, \tau_r$, the permutation $\rho_i$ got applied to the program cards if and only if $x_\ell = 1$ (otherwise the order of the program cards does not change).

Alice and Bob repeat this procedure for each layer of the branching program. 
Let $\alpha$ be the card representing the accepting vertex $v_\text{acc}$ and $\beta$ be the card representing the rejecting vertex $v_\text{rej}$ at the end of the simulation.
The cards $\alpha, \beta$ represent the output of the program.
If the input is accepted, then the accepting vertex $v_\text{acc}$ is active at the end of the simulation and thus the card $\alpha$ has suit $\co$ and the card $\beta$ has suit $\cz$.
Thus, the cards $\alpha, \beta$ represent 1.
On the other hand, if the input is rejected, then the rejecting vertex $v_\text{rej}$ is active.
Thus, the cards $\alpha, \beta$ have suits $\cz$ and $\co$, respectively, and they represent 0.


The protocol $\Pi$ is clearly {\CS}-protocol as the players only sequentially apply the copying protocol $\Gamma_1$ and the swapping protocol $\Gamma_2$ to the program cards.
We claim the simulation protocol $\Pi$ is secure.
Both protocols $\Gamma_1$ and $\Gamma_2$ are secure.
The only helping cards which are used during the whole run of the protocol $\Pi$ are program cards which are placed face-up from the deck and then turned face-down for the rest of the protocol.

\end{proof}

Now we prove the opposite inclusion of Theorem~\ref{thm:OneCard}.

\begin{theorem}
\label{thm:OneCard2}
$\CS \subseteq \BP$.
\end{theorem}
\begin{proof}
Consider a family of functions $\{ f_n: \{0,1\}^n \times \{0,1\}^n \to \{0,1\} \}_{n\ge 0}$ for which we have a sequence of secure read-only oblivious protocols, one for each $n$, which are of length polynomial in $n$, with deck size $s$ and work space size $2s$.
Let $c$ be the number of different cards used by the protocol. At any moment, the work space can be in at most $(2c+1)^{2s}$ different states which we call the internal states of the protocol.
For any $n\ge 1$ we will build a width-$(2c+1)^{2s}$ branching program of the same length $T_n$ as the protocol for $f_n$.
Each layer of the branching program consists of vertices where each vertex corresponds to one internal state of the protocol. We need to define edges between the layers of the branching program.

Let $v$ be a vertex at layer $t \in \{0,\dots, T_n-1\}$. It corresponds to some internal state which in turn determines a visible state that together with $t$ determines the action taken by the protocol at such a state. We define the edges based on the type of that action.
If the action is a move of a card from some input position into the work space then node $v$ queries the value of the corresponding 
input variable and the outgoing edges lead to nodes corresponding to internal states that reflect a move of the card into the work space.
If the action is a shuffle operation then node $v$ queries variable $x_1$ and irrespective of its value both outgoing edges go to the node in the next
layer corresponding to the internal state obtained by applying one of the allowed permutations. (The particular choice of the permutation does not matter.) Similarly, for a move of a card within the work space or out from the work space, the edges will go into a node that reflects the internal state after the move. In the last layer, we designate vertices that correspond to accepting states of the protocol as accepting all other nodes will be rejecting.

It should be clear from the construction that the resulting branching program computes $f_n$ and has the required properties.
\end{proof}

Theorem~\ref{thm:OneCard} is a corollary of Theorems~\ref{thm:OneCard1} and~\ref{thm:OneCard2} and Barrington's theorem (Theorem~\ref{thm:Barrington}).

\section{Simulating Turing Machines}

In this section, we will look at computation that obliviously and securely computes on committed inputs in 2-card representation.
The exact split of the input between Alice and Bob is irrelevant in this section so we assume that the total length of the input is $n$ bits.
The protocols are expected to preserve the committed inputs: Although they may be allowed to modify the committed input during the computation, by the end of the computation the committed input must be restored to its original form. The protocols do not leak any information about the committed inputs except for what can be derived from the output cards if they are inspected. The protocols can be carried out by either player. To guarantee robustness and security shuffle operations should be always done by both players. (We use only uniformly random shuffle and random cyclic shift so performing them twice does not change their output distribution.)

\begin{theorem}\label{thm:s-space}
Let $s(n)\ge \log n$ be a non-decreasing function. Let $f$ be a function computable by a Turing machine in space $s(n)$.
Then $f$ is in $O(s(n))$-\CS.
\end{theorem}

Let $\SEL:\{0,1\}^n \to \{0,1\}$ be the function such that $\SEL(c,b,a)=a$ if $c=0$ and $\SEL(c,b,a)=b$ otherwise.

\begin{proof}
We describe the algorithm for the protocol in high-level form and leave details of the construction to the interested reader.
Let $f$ be computable by a Turing machine $M$. Without loss of generality, we assume $M$ uses a binary alphabet, on inputs of length $n$ 
it uses work space exactly $s(n)$ bits and computes for $t(n)$ steps. The output of $M$ is determined by the first bit of its work tape.
We will simulate the computation of $M$ step by step. 

The protocol will use $8 s(n) + 4\log n + O(1)$ auxiliary cards. Two blocks $w$ and $w'$ will represent $2 s(n)$ bits, each, and two blocks
$p$ and $j$ will represent $\log n$ bits each. In addition to that there is a block $q$ of $O(1)$ bits, and some additional auxiliary bits. 
We need a constant number of positions to be empty.
All the bits are encoded in 2-card representation. The block $w$ represents the content of the work tape of $M$, 2 bits per tape cell, where the second bit
indicates the presence of the work tape head on that particular tape cell. The block $p$ encodes in binary the current position of the input head of $M$.
The block $q$ encodes the internal state of $M$. 

The protocol simulates one step of $M$ as follows: first, it determines the value $b$ of the bit scanned by the input head, then it calculates into $w'$ the content of the work tape of $M$ after this step. Then it updates the internal state, the input head position $p$, and switches $w'$ and $w$.

To determine $b$ the protocol looks at each input bit $x_i$ one by one and records the one that has an index corresponding to $p$. Set $b$ to 0.
For $i=1,\dots,n$, the protocol copies $x_i$ into some work space $b'$, it sets $j$ to represent $i$, obliviously compares $p$ 
and $j$ while recording the result into $c$. (Comparing bit strings can be done by an $\NC$ circuit so there is an oblivious protocol for that of poly-logarithmic length.) 
From $c$, $b$ and $b'$ we can calculate the new value of $b$ by evaluating $\SEL(c,b,b')$. This can be done obliviously.
After processing all the input bits, $b$ has the value of the currently scanned input bit.

Now, we can determine $w'$, the content of the work tape of $M$ after this step of the computation. We compute $w'$ cell by cell. The value of each cell is a function of the input bit $b$, $M$'s state $q$, and the previous content of the cell in $w$ together with the content of adjacent cells. 
Hence, the value of each bit of $w'$ is a function of constantly many bits and can be computed obliviously. 

After computing $w'$, we can also calculate $d$, the direction in which the input head of $M$ should move, and the new state $q'$ of $M$.
This can be done by scanning $w$ for the work tape position, and recording the relevant information for $q'$ and $d$ when we pass over the current work cell similarly to determining the value of the input bit $b$.

From $p$ and $d$, we obliviously calculate the next position of the input head into $j$. (Each bit of $j$ can be computed by an $\NC$ circuit from $p$ and $d$.) 
Finally, we switch the contents of $w$ and $w'$, $p$ and $j$, and $q$ and $q'$.

We repeat this procedure $t(n)$ times. In the end, the first bit of $w$ indicates the output of $M$.

As each step of the computation can be implemented securely, obliviously, and robustly, we obtain a secure, oblivious, and robust protocol for $f$
that uses $O(s(n))$ work space and $O(s(n))$ auxiliary cards.
\end{proof}

As the card-based protocols allow for non-uniformity by protocols using $O(\log n)$ work space we can simulate not only log-space Turing machines
but also polynomial-size branching programs (the \emph{non-uniform log-space}). The above proof can be extended to Turing machines taking advice:
the protocol can provide the advice bit by bit during the phase when the input is scanned bit by bit to determine $b$.
(We assume that the advice is provided to the Turing machine on bit positions with index $>n$. For those positions instead of copying the non-existent input bits, the protocol hardwires the appropriate bit into $b'$. As the advice is the same for each input, this can be done publicly.)

By essentially the same proof as Theorem~\ref{thm:OneCard2} we can obtain a simulation of oblivious, read-only secure protocols that use a logarithmic amount of work space by branching programs of polynomial size. Let $O(\log n)$-$\CS = \bigcup_k (k+k\log n)$-$\CS$. We get:

\begin{theorem}
The class of functions computable by polynomial-size branching programs equals to $O(\log n)$-\CS.
\end{theorem}

\subsection{Read-write Protocols}

So far we have looked only at read-only protocols. If we remove the condition to be read-only we get a potentially larger class of functions computable
by such protocols. When the protocol is not read-only, we still require the protocol to restore its input into the original state by the end of the computation. We also require the protocol to be secure so not to leak any information about the input except for what is implied by the protocol output cards.

We give examples of functions that can be computed by protocols modifying their input which we conjecture are outside of the read-only protocol class with similar bound on the work space. Proving this conjecture would amount to separating $\NC$ from log-space, a major open problem in complexity theory.

Let $s(n)$ be a non-decreasing function such that $\log n \le s(n) \le n/2\log^* n$.
Let $g:\{0,1\}^n \to \{0,1\}$ be in $\NC$, and $h:\{0,1\}^{n-s(n) \log^* n}  \to \{0,1\}$  be a function computable by a Turing machine in space $O(s(n))$ and polynomial time. Define $f:\{0,1\}^n \to \{0,1\}$ as follows:
$$
f(x) =
\left\{
	\begin{array}{ll}
		g(x)  & \mbox{if } x_{n+1-s(n)\log^* n}\cdots x_n \ne 0\cdots 0,   \\
		h(x_1\cdots x_{n-s(n)\log^* n}) & \mbox{otherwise.} 
	\end{array}
\right.
$$

\begin{theorem}\label{thm:RWSim}
The function $f$ defined above is computable by secure robust oblivious protocols of polynomial length that use a constant amount of work space.
\end{theorem}

\begin{proof}
The protocol for $f$ proceeds as follows. It first computes the OR of the input bits $x_{n+1-s(n)\log^* n}\cdots x_n$
using a protocol for $\NC$ functions where the output $c$ of the protocol is encoded in 2-card representation in its work space.
Then it computes the value $g$ of $g(x)$ encoded in 2-card representation in the work space. Finally, it uses the cards representing
input bits  $x_{n+1-s(n)\log^* n}\cdots x_n$ to simulate the computation of a Turing machine for $h$ as in the proof of Theorem~\ref{thm:s-space}.
The simulation is done so that if $c=1$ then nothing is done to the input (the simulation is vacuous) and if $c=0$ the simulation is really happening. The simulation uses the input bits  $x_{n+1-s(n)\log^* n}\cdots x_n$ to store $w,w',p$ and $j$ (from the simulation), 
everything else is done in the actual work space of constant size.
Whenever the simulation wants to write some value $a$ into an input position used for the simulation,
it copies the value into the work space, it copies there the current value $d$ of the destination position, computes $\SEL(c,a,d)$ and replaces cards in the destination by the output of $\SEL(c,a,d)$. (Hence, if $c=0$ nothing has happened.) Reading a value can be done by copying the particular bit into the work space and then working with the copied cards. This way the input is undisturbed if $c=1$ and it is overwritten if $c=0$.
At the end of simulation, the protocol copies the output bit $h$, which is the first bit of $w$ into the actual work space, and writes value 0 to all input bits $x_{n+1-s(n)\log^* n}\cdots x_n$ conditionally on $c=0$.
(Bit values read from the storage, taking part in the vacuous computation, will be the same throughout the computation. So the simulations of the $\NC$ circuits implementing various steps of the computation will be secure.)

Finally, the protocol computes $\SEL(c,h,g)$ which is the output of the protocol. All parts of the protocol can be done securely and obliviously.
(This is true also when $c=0$ and the simulation of the Turing machine is bogus.) The protocol restores its committed input by the end of the computation.
\end{proof}

One can also use the technique of catalytic computation to construct protocols for functions not know to be in $\NC$. Buhrman et al. \cite{BCKLS,Cleve-phd} show how to use memory which contains
some information for computation while restoring the memory to its original content by the end of the computation. For example, they can solve the connectivity on directed graphs this way, the problem $\CONN(G)$: Given an $n\times n$ adjacency matrix of a directed graph $G$ decide whether there is a path from vertex $1$ to vertex $n$. They present a polynomial-size program for $\CONN(G)$, which uses $3n^2+1$ work registers and $n^2$ input registers, each holding one bit of information. The program consists of instruction of the form
\[
 r_i \leftarrow r_i \oplus u \cdot v,
\]
where $u$ and $v$ are arbitrary registers different from $r_i$ or constants 0 and 1. The program is oblivious, so it is a straight line program consisting of such instructions.
The program has the property that all registers are guaranteed to have the initial value by the end of the computation except for one specified work register which contains the output value. It is straightforward to implement each such an instruction by secure and robust protocol since the instructions are computable in $\NC$.  

This allows to design an oblivious, secure protocol of polynomial length with constant work space for a function $f' : \{0,1\}^{4n^2} \to \{0,1\}$ that is defined as: $f'(G_1,G_2,G_3,G_4)=1$ if and only if from the vertex $1$ we can reach the vertex $n$ in each of the graphs represented by adjacency matrices $G_1,G_2,G_3$ and $G_4$. Such a function is unlikely to be contained in $\NC$, as $\CONN(G)$ is known to be complete for \emph{nondeterministic} log-space computation. Hence, it is unlikely that protocols that are allowed to modify their input could be simulated by read-only protocols using similar resources.

\section{More Efficient Input Encodings}

\subsection{1-Card Encoding}

In this section, we consider other ways how Alice and Bob can commit their input which use fewer cards.
The first natural encoding is to represent each bit $1$ by face-down card $\co$ and bit $0$ by $\cz$. These cards would be stored 
in front of the players in input positions $1,\dots,2n$.
Whenever the players want to operate with the committed bit they need to extend it to 2-card representation.

There are two ways we know how to do it. Niemi and Renvall \cite{niemi97} gave a protocol that is able to extend the bit without knowing its value.
However, there is a small probability of leaking the value of the bit being extended. The probability is inversely proportional to the number of cards used for the protocol. Hence, one would need a large number of helping cards in order to make sure that the probability of leaking information is negligible. That would erase any savings from the 1-card representation. 

The other way which we use here is to allow the player who owns the particular input bit to extend it using a designated deck of two face-down cards $\cz$ and $\co$. Once the bit is extended it can be copied by the protocol from Theorem~\ref{thm:CopyBit}, the first card of the first copy can be put back in the input position, the second card can be put back into the auxiliary deck, and the second copy can be used for further computation. The auxiliary deck containing the same cards as earlier should be shuffled by both players at the end of this procedure. 
The protocol is robust since a player cheating by extending the input bit by a wrong card will be caught in Step 4 of the copying protocol.

For this procedure, we need to augment our set of actions by the action of extending a bit by a complementary card from a designated deck. This action can be performed by shuffling the auxiliary deck, then peeking at the value of the card we are extending, and selecting the complementary card by peeking at each card in the deck. 

With this operation in mind, we need to extend the definition of protocol security. We say a protocol is \emph{secure from Alice} if for any pair of inputs $(x,y)$ and $(x,y')$ to Alice and Bob, the distribution of the sequence of visible states of the protocol 
together with the sequence of cards seen by Alice while peeking at them during the extension action on inputs $(x,y)$ and $(x,y')$ is the same.
Similarly, the protocol is \emph{secure from Bob} if for any pair of inputs $(x,y)$ and $(x',y)$ to Alice and Bob, the distribution of visible states and cards peeked at by Bob will be the same on both inputs $(x,y)$ and $(x',y)$. The protocol is secure if it is secure from both Alice and Bob.

Using the extension action we can perform all read-only protocols that used the 2-card bit commitment of inputs even for inputs committed in 1-card representation. They will be secure as long as the player performing each extension is the owner of the input bit as seeing his/her input bits does not affect the security definition. Hence, the power of the model stays essentially the same with this modification.

Security becomes more of an issue for protocols that are allowed to modify their inputs. Yet, we can prove a result similar to
Theorem~\ref{thm:RWSim} for slightly modified function $f'$. Let $g:\{0,1\}^n \to \{0,1\}$ be in $\NC$, and $h:\{0,1\}^{n-s(n) \log^* n}  \to \{0,1\}$  be a function computable by a Turing machine in space $O(s(n))$ and polynomial time, where $\log n \le s(n) \le n/2\log^* n$.. Define $f':\{0,1\}^n \to \{0,1\}$ as follows:
$$
f'(x) =
\left\{
	\begin{array}{ll}
		g(x)  & \mbox{if } x_{n+1-s(n)\log^* n}\cdots x_n \ne 0101\cdots 01,   \\
		h(x_1\cdots x_{n-s(n)\log^* n}) & \mbox{otherwise.} 
	\end{array}
\right.
$$

We assume $s(n)\log^* n$ is even, and the first half of $x_{n+1-s(n)\log^* n}\cdots x_n$ is held by Alice and the other half by Bob. 
The other bits can be split between the players arbitrarily.

\begin{restatable}{theorem}{rwsim}
\label{thm:RWSim1}
The function $f'$ defined above is computable by secure robust oblivious protocols of polynomial length that use a constant amount of work space and 1-card encoding of input bits.
\end{restatable}
\begin{proof}
The protocol for $f'$ proceeds similarly to the one in Theorem~\ref{thm:RWSim}. 
It first verifies whether the input bits $x_{n+1-s(n)\log^* n}\cdots x_n$ differ from $0101\cdots 01$ (assuming their number is even)
using protocol for $\NC$ functions. The output $c$ of the verification is encoded in 2-card representation in the work space.
Then the protocol computes the value $g$ of $g(x)$ encoded in 2-card representation in the work space. Up until this point, we use the protocol described above to extend input bits into 2-card representations by the player who owns the input bit.

Now we want to use the cards representing input bits  $x_{n+1-s(n)\log^* n}\cdots x_n$ to simulate computation of a Turing machine $M$ for computing $h=h(x_1,\dots,x_{n- s(n)\log^* n})$ as in the previous proofs.
We will use these input cards for storage when $c=0$ and when $c=1$ we will keep them intact. In the former case, we will eventually reset the input bits/cards to the initial state.
Let $I$ be positions of cards which represent bits $x_1\cdots x_{n - s(n)\log^* n}$ at the beginning of the protocol.
Thus, the cards on the positions $I$ represent (in the 1-card encoding) the input for $M$.
These cards will be in a read-only regime during the whole computation.
Let $J$ be the positions of cards representing $x_{n+1-s(n)\log^* n}\cdots x_n$.
The cards on $J$ will represent in 2-card encoding the content of tapes of $M$ during the computation.
Thus, $|J| = s(n)\log^* n$ but the cards on $J$ will represent $\frac{1}{2}s(n)\log^* n$ bits.
As $x_{n+1-s(n)\log^* n}\cdots x_n = 0101\cdots 01$ (if $c = 0$), the cards on $J$ represent $\frac{1}{2}s(n)\log^* n$ zeros at the beginning of the protocol.

The simulation proceeds in a similar way as the simulation in proof of Theorem~\ref{thm:RWSim}.
For reading bits from $I$ that encode the input bits to $M$ we use the same 1-card extension protocol as above to copy them into work-space.
Now, we need procedures that will read and write bits in 2-card representations from positions $J$.
However, if $c = 1$ some two consecutive positions would not represent a bit correctly (the two cards on them would have the same suit).
The read/write procedures need to work and be secure even in this case.
The players cannot simply inspect the cards on positions $J$ because they may represent some intermediate results of the computation.

First, we describe how to read a bit $b$ encoded in $J$. We want to create a 2-card representation of bit $b$ in work space if $c = 0$ or a valid 2-card representation of some bit if $c = 1$.
The bit $b$ is represented by 2 cards $\alpha, \beta$ on positions in $J$.
Note that $\alpha$ and $\beta$ can be of the same suit if $c = 1$.
Suppose Alice owns the positions in $J$ representing the bit $b$, the case of Bob's positions is symmetric.
First, Alice will add complementary cards to $\alpha$ and $\beta$ as follows.
Alice prepares the sequence: $\cz\cz \alpha \cz\co \beta$, then she turns the second and fifth card face-down to get a sequence 
\[
\cz ? \alpha \cz ? \beta.
\]
Bob shuffles the six cards cyclically at random. 
Now, Alice extends the card preceding each $\cz$ that is face-up (cards $\alpha$ and $\beta$) by a complementary card to the right taken from an auxiliary deck.
Alice sees the suit of the cards $\alpha$ and $\beta$ during this action but she does not know their actual order.
Bob shuffles the cards cyclically again and turns face-up the cards following the two $\cz$ that are face-up. Now, by a cyclic shift, they rearrange the cards so that they look like 
\[
\cz\cz \alpha \overline{\alpha} \cz\co \beta \overline{\beta},
\]
where $\overline{\alpha}$ and $\overline{\beta}$ are cards complementary to $\alpha$ and $\beta$, respectively.
They can copy each of the card pairs $\alpha,\overline{\alpha}$ and $\beta,\overline{\beta}$ to verify that Alice used complementary cards and the pairs $\alpha,\overline{\alpha}$ and $\beta,\overline{\beta}$ indeed represent two bits in 2-card encodings.

By following this protocol, Alice learns whether the two cards $\alpha,\beta$ have the same suit or are distinct. 
If the suits are the same she also learns the suit.
However, in that situation $c=1$ and she already knew all this information. 
In the later case, she knows that the bits at those input positions are distinct, but she knew that already. 
She does not learn their relative order because of the shuffle by Bob. 
Thus, she does not know whether they were altered since the beginning of the protocol or not.

To finish the read procedure, Alice copies the pair $\alpha,\overline{\alpha}$ (by the protocol of Theorem~\ref{thm:CopyBit}) to get two copies represented by cards $\alpha',\overline{\alpha'},\alpha'',\overline{\alpha''}$.
She returns the cards $\alpha''$ and $\beta$ back to the positions in $J$ from which she moved the cards $\alpha$ and $\beta$ at the beginning.
The cards $\alpha'$ and $\overline{\alpha'}$ are used further in the computation.
Other cards ($\overline{\alpha''}, \overline{\beta}$) are moved back to the auxiliary deck.
If the cards $\alpha$ and $\beta$ have different suits then the cards $\alpha',\overline{\alpha'}$ represent the bit $b$, as the card $\alpha$ and $\alpha'$ have the same suit.
If the cards $\alpha$ and $\beta$ have the same suit then the cards $\alpha',\overline{\alpha'}$ are a valid representation of some bit $b'$.
However, in that case $c = 1$ and the value of $b'$ is irrelevant for the computation.
Bit values read from the storage, taking part in the bogus computation, will be consistent throughout the computation. 
Thus, the simulations of the $\NC$ circuits implementing various steps of the computation are secure.

Now, we describe how to store a bit in 2-card encodings on to some positions in $J$.
Again, suppose we want to store a bit $b$ on to positions owned by Alice and occupied by cards $\gamma_1$ and $\gamma_2$.
Let $\alpha$ and $\beta$ be cards representing $b$.
We want a procedure that will do the following.
If $c = 0$ then the cards $\gamma_1$ and $\gamma_2$ are replaced by cards of the same suits as $\alpha$ and $\beta$, respectively.
Otherwise, if $c = 1$ then the new cards need to have the same suits as $\gamma_1$ and $\gamma_2$.
First, Alice will add complementary cards to $\gamma_1$ and $\gamma_2$ to get a sequence $\gamma_1,\overline{\gamma_1},\gamma_2,\overline{\gamma_2}$ (she proceeds in the same was as in the read procedure above).
Let $d_1$ and $d_2$ be bits represented by $\gamma_1,\overline{\gamma_1}$ and $\gamma_2,\overline{\gamma_2}$, respectively.
She creates two copies of $b$, negates the second one, and computes $a_1 = \SEL(c,b,d_1)$ and $a_2 = \SEL(c,1-b,d_2)$.
Let $\delta_1,\overline{\delta_1}$ and $\delta_2,\overline{\delta_2}$ be cards representing $a_1$ and $a_2$ respectively.
If $c = 0$ then the cards $\delta_1,\delta_2$ represent the bit $b$.
If $c = 1$ then the cards $\delta_1$ and $\delta_2$ have the same suits as the cards $\gamma_1$ and $\gamma_2$, respectively.
Thus, Alice moves the cards $\delta_1$ and $\delta_2$ into the positions of the cards $\gamma_1$ and $\gamma_2$ and moves the rest of the cards to the deck.

To avoid leakage of information from the way Alice picks the cards from the auxiliary deck when picking a card of a particular suit, she proceeds as follows.
She knows how many cards of that suit are in the deck. 
Thus, she shuffles the cards at random and then proceeds left to right to pick one of the cards of that suit uniformly at random. 
To achieve that she picks each card of the desired suit with probability $1/(k+1)$, where $k$ is the number of unseen cards of the desired suit still in the deck. 
This process guarantees that Alice will pick a card from a completely random position.

In this way the protocol can use $x_{n+1-s(n)\log^* n}\cdots x_n$ to compute $h=h(x)$. 
After obtaining value $h$ it outputs $\SEL(c,h,g)$.
\end{proof}

Hence, also in the case of the 1-card representation of the input one can take advantage of the input cards to compute functions that seem unattainable 
with read-only protocols.

\subsection{1/2-Card Encoding}

In the 1/2-card encoding we represent value 1 by either $\co \ce$ or $\ce \cz$, and value 0 by either  $\cz \ce$ or $\ce \co$. 
Here $\ce$ represents an empty bit position.
To commit her input Alice picks $n/2$ of her input bits, and for those input bits, she leaves the empty spot $\ce$ in the position of $\co$,
for the remaining bits she leaves the empty spot in place of $\cz$ (in the 2-card encoding of the bit). This way, she uses exactly $n/2$ cards $\cz$ and $\co$ to commit her input.
It is easy to verify that for each bit there is exactly $1/2$ probability that the missing card will be on the left. Hence, the positions of the missing cards do not leak any information about her input. Bob proceeds in the same way to commit his input.

After committing their inputs they can run any read-only protocol similar to the case of 1-card encoding. Whenever an input bit is needed it is copied into a 2-card representation by essentially the same protocol as in the case of 1-card encoding.

This means that we need only $n+O(1)$ cards to compute any $\NC$ function on $n$-bit inputs. 

We do not know how to implement protocols which could modify their inputs. Modifying an input bit would require either picking the empty spot in the representation at random (which could lead to using substantially more cards of each type) or reusing the cards that are there. In the latter case we do not know how to do it without leaking information.

\bibliography{main}

\begin{thebibliography}{10}

\bibitem{abe18}
Yuta Abe, Yu-ichi Hayashi, Takaaki Mizuki, and Hideaki Sone.
\newblock Five-card and protocol in committed format using only practical
  shuffles.
\newblock In {\em Proceedings of the 5th ACM on ASIA Public-Key Cryptography
  Workshop}, APKC '18, page 3–8, New York, NY, USA, 2018. Association for
  Computing Machinery.
\newblock URL: \url{https://doi.org/10.1145/3197507.3197510}, \href
  {http://dx.doi.org/10.1145/3197507.3197510}
  {\path{doi:10.1145/3197507.3197510}}.

\bibitem{barrington89}
David Barrington.
\newblock Bounded-width polynomial-size branching programs recognize exactly
  those languages in nc1.
\newblock {\em Journal of Computer and System Sciences}, 38:150--164, 02 1989.
\newblock \href {http://dx.doi.org/10.1016/0022-0000(89)90037-8}
  {\path{doi:10.1016/0022-0000(89)90037-8}}.

\bibitem{BCKLS}
Harry Buhrman, Richard Cleve, Michal Kouck\'{y}, Bruno Loff, and Florian
  Speelman.
\newblock Computing with a full memory: Catalytic space.
\newblock In {\em Proceedings of the Forty-Sixth Annual ACM Symposium on Theory
  of Computing}, STOC '14, page 857–866, New York, NY, USA, 2014. Association
  for Computing Machinery.
\newblock URL: \url{https://doi.org/10.1145/2591796.2591874}, \href
  {http://dx.doi.org/10.1145/2591796.2591874}
  {\path{doi:10.1145/2591796.2591874}}.

\bibitem{Cleve-phd}
R.~E. Cleve.
\newblock {\em Methodologies for Designing Block Ciphers and Cryptographic
  Protocols}.
\newblock PhD thesis, CAN, 1989.

\bibitem{crepeau94}
Claude Cr{\'e}peau and Joe Kilian.
\newblock Discreet solitary games.
\newblock In Douglas~R. Stinson, editor, {\em Advances in Cryptology ---
  CRYPTO' 93}, pages 319--330, Berlin, Heidelberg, 1994. Springer Berlin
  Heidelberg.

\bibitem{boer90}
Bert den Boer.
\newblock More efficient match-making and satisfiability the five card trick.
\newblock In Jean-Jacques Quisquater and Joos Vandewalle, editors, {\em
  Advances in Cryptology --- EUROCRYPT '89}, pages 208--217, Berlin,
  Heidelberg, 1990. Springer Berlin Heidelberg.

\bibitem{danny17}
Danny Francis, Syarifah~Ruqayyah Aljunid, Takuya Nishida, Yu-ichi Hayashi,
  Takaaki Mizuki, and Hideaki Sone.
\newblock Necessary and sufficient numbers of cards for securely computing
  two-bit output functions.
\newblock In Rapha{\"e}l C.-W. Phan and Moti Yung, editors, {\em Paradigms in
  Cryptology -- Mycrypt 2016. Malicious and Exploratory Cryptology}, pages
  193--211, Cham, 2017. Springer International Publishing.

\bibitem{kastner17}
Julia Kastner, Alexander Koch, Stefan Walzer, Daiki Miyahara, Yu-ichi Hayashi,
  Takaaki Mizuki, and Hideaki Sone.
\newblock The minimum number of cards in practical card-based protocols.
\newblock In Tsuyoshi Takagi and Thomas Peyrin, editors, {\em Advances in
  Cryptology -- ASIACRYPT 2017}, pages 126--155, Cham, 2017. Springer
  International Publishing.

\bibitem{koch18}
Alexander Koch.
\newblock The landscape of optimal card-based protocols.
\newblock Cryptology ePrint Archive, Report 2018/951, 2018.
\newblock url{https://eprint.iacr.org/2018/951}.

\bibitem{koch21}
Alexander Koch and Stefan Walzer.
\newblock Foundations for actively secure card-based cryptography.
\newblock In Martin Farach{-}Colton, Giuseppe Prencipe, and Ryuhei Uehara,
  editors, {\em 10th International Conference on Fun with Algorithms, {FUN}
  2021, May 30 to June 1, 2021, Favignana Island, Sicily, Italy}, volume 157 of
  {\em LIPIcs}, pages 17:1--17:23. Schloss Dagstuhl - Leibniz-Zentrum f{\"{u}}r
  Informatik, 2021.
\newblock URL: \url{https://doi.org/10.4230/LIPIcs.FUN.2021.17}, \href
  {http://dx.doi.org/10.4230/LIPIcs.FUN.2021.17}
  {\path{doi:10.4230/LIPIcs.FUN.2021.17}}.

\bibitem{koch15}
Alexander Koch, Stefan Walzer, and Kevin H{\"a}rtel.
\newblock Card-based cryptographic protocols using a minimal number of cards.
\newblock In Tetsu Iwata and Jung~Hee Cheon, editors, {\em Advances in
  Cryptology -- ASIACRYPT 2015}, pages 783--807, Berlin, Heidelberg, 2015.
  Springer Berlin Heidelberg.

\bibitem{mizuki16}
Takaaki Mizuki.
\newblock Card-based protocols for securely computing the conjunction of
  multiple variables.
\newblock {\em Theor. Comput. Sci.}, 622(C):34–44, April 2016.
\newblock URL: \url{https://doi.org/10.1016/j.tcs.2016.01.039}, \href
  {http://dx.doi.org/10.1016/j.tcs.2016.01.039}
  {\path{doi:10.1016/j.tcs.2016.01.039}}.

\bibitem{mizuki12}
Takaaki Mizuki, Michihito Kumamoto, and Hideaki Sone.
\newblock The five-card trick can be done with four cards.
\newblock In Xiaoyun Wang and Kazue Sako, editors, {\em Advances in Cryptology
  -- ASIACRYPT 2012}, pages 598--606, Berlin, Heidelberg, 2012. Springer Berlin
  Heidelberg.

\bibitem{mizuki14}
Takaaki Mizuki and Hiroki Shizuya.
\newblock A formalization of card-based cryptographic protocols via abstract
  machine.
\newblock {\em Int. J. Inf. Secur.}, 13(1):15–23, February 2014.
\newblock URL: \url{https://doi.org/10.1007/s10207-013-0219-4}, \href
  {http://dx.doi.org/10.1007/s10207-013-0219-4}
  {\path{doi:10.1007/s10207-013-0219-4}}.

\bibitem{mizuki14a}
Takaaki Mizuki and Hiroki Shizuya.
\newblock Practical card-based cryptography.
\newblock In Alfredo Ferro, Fabrizio Luccio, and Peter Widmayer, editors, {\em
  Fun with Algorithms}, pages 313--324, Cham, 2014. Springer International
  Publishing.

\bibitem{mizuki09}
Takaaki Mizuki and Hideaki Sone.
\newblock Six-card secure and and four-card secure xor.
\newblock In Xiaotie Deng, John~E. Hopcroft, and Jinyun Xue, editors, {\em
  Frontiers in Algorithmics}, pages 358--369, Berlin, Heidelberg, 2009.
  Springer Berlin Heidelberg.

\bibitem{niemi97}
Valtteri Niemi and Ari Renvall.
\newblock Secure multiparty computations without computers.
\newblock Technical report, 1997.

\bibitem{nishida15}
Takuya Nishida, Yu-ichi Hayashi, Takaaki Mizuki, and Hideaki Sone.
\newblock Card-based protocols for any boolean function.
\newblock In Rahul Jain, Sanjay Jain, and Frank Stephan, editors, {\em Theory
  and Applications of Models of Computation}, pages 110--121, Cham, 2015.
  Springer International Publishing.

\bibitem{nishimura16}
Akihiro Nishimura, Yu-ichi Hayashi, Takaaki Mizuki, and Hideaki Sone.
\newblock An implementation of non-uniform shuffle for secure multi-party
  computation.
\newblock In {\em Proceedings of the 3rd ACM International Workshop on ASIA
  Public-Key Cryptography}, AsiaPKC '16, page 49–55, New York, NY, USA, 2016.
  Association for Computing Machinery.
\newblock URL: \url{https://doi.org/10.1145/2898420.2898425}, \href
  {http://dx.doi.org/10.1145/2898420.2898425}
  {\path{doi:10.1145/2898420.2898425}}.

\bibitem{stiglic01}
Anton Stiglic.
\newblock Computations with a deck of cards.
\newblock {\em Theor. Comput. Sci.}, 259(1):671–678, May 2001.

\end{thebibliography}
\end{document}